\documentclass[12pt] {article}

\usepackage{latexsym}
\usepackage[dvips]{graphics,color}
\usepackage{graphpap}
\usepackage{alltt}
\usepackage{graphicx}
\usepackage{amsmath}
\usepackage{amsfonts}
\usepackage{psfrag}
\newtheorem{theorem}{Theorem}[section]
\newtheorem{lemma}[theorem]{Lemma}
\newtheorem{proposition}[theorem]{Proposition}
\newtheorem{proof}[theorem]{Proof}
\newtheorem{definition}[theorem]{Definition}
\newcommand{\beq}{\begin{equation}}
\newcommand{\feq}[1]{\label{#1} \end{equation}}
\newcommand{\beqr}{\begin{eqnarray}}
\newcommand{\feqr}{\end{eqnarray}}
\def\non{\nonumber}
\def\noi{\noindent}

\definecolor{red}{rgb}{1,0,0}
\def\pre#1#2{Phys. Rev. {\bf{E#1}}, #2}

\def\ap#1#2{Ann. of Phys. {\bf{#1}}, #2}

\def\tams#1#2{Trans. Amer. Math. Soc. {\bf{A#1}}, #2}

\begin{document}

\begin{center}


{\Large \bf The weakly coupled fractional one-dimensional Schr\"{o}dinger operator with index $\bf 1<\alpha \leq 2$}\\
[4mm]

\large{Agapitos N. Hatzinikitas} \\ [5mm]

{\small University of Aegean, \\
School of Sciences, \\
Department of Mathematics, \\
Karlovasi, 83200\\
Samos Greece \\
E-mail: ahatz@aegean.gr}\\ [5mm]

\end{center}
\begin{abstract}
We study fundamental properties of the fractional, one-dimensional Weyl operator $\hat{\mathcal{P}}^{\alpha}$ densely defined on the Hilbert space $\mathcal{H}=L^2({\mathbb R},dx)$ and determine the asymptotic behaviour of both the free Green's function and its variation with respect to energy for bound states. In the sequel we specify the Birman-Schwinger representation for the Schr\"{o}dinger operator $K_{\alpha}\hat{\mathcal{P}}^{\alpha}-g|\hat{V}|$ and extract the finite-rank portion which is essential for the asymptotic expansion of the ground state. Finally, we determine necessary and sufficient conditions for there to be a bound state for small coupling constant $g$.
\end{abstract}

\section{Introduction}
\label{sec0}
In the present work we consider an infinite complex and separable Hilbert space $\mathcal{H}=L^2({\mathbb R},dx)$ with inner product defined by \ref{sec1 : eq1} and the fractional Weyl operator $\hat{\mathcal{P}}^{\alpha}$ defined by \ref{sec1 : eq2}. This nonlocal operator is the infinitesimal generator of time translations for symmetric $\alpha$-stable L\'{e}vy processes \cite{Ref1a} and form a strongly continous semigroup \cite{Ref7,Ref8}. We prove the properties: linearity, unboundeness, symmetricity and self-adjointness for $\hat{\mathcal{P}}^{\alpha}$. 
\par A key r\^{o}le to the Birman-Schwinger representation plays the free Green's function (FGF) which obeys \ref{sec2 : eq6}. Inserting a contour representation for the function $e^{-u^{\alpha}}$, present in the Fourier-Laplace transform of FGF, and performing the integration on the complex $s$-plane we determine the asymptotic behaviour of FGF near the points $z=0$ and $z=\infty$. The behaviour is only investigated for irrational values of the index-$\alpha$ since for rational values the corresponding expressions are $\alpha$-sensitive. In the infinity, the FGF is bounded and the operator pocesses physically acceptable bound states. The FGF in the neighbourhood of $z=0$ is decomposed into three terms \ref{sec2 : eq20a}, in the zero energy limit: a singular part, a constant energy independent part and a regular power series part. As an application, in the $\alpha \uparrow 2$ limit, we recover the known one-dimensional FGF. We also study the asymptotic behaviour of $\partial G_{\alpha}/\partial \kappa$ near the point $z=0$ and observe that it splits into only a singular and a regular part in the $\kappa \downarrow 0$ limit. In this case the energy inpendent term appears only for $\alpha=2$. 
\par Next we consider the Schr\"{o}dinger operator \ref{sec3 : eq01} in one space dimension with a multiplication operator in $L^2({\mathbb R})$ which is assumed to be continuous, real, strictly negative function and tends to zero sufficiently fast as $|x|\rightarrow \infty$. The eigenvalue problem \ref{sec3 : eq03} for bound states, using the Birman-Schwinger transformation, is  transformed into an equivalent one with integral operator having kernel expressed in terms of the Green's function \ref{sec2 : eq9}. Under a supplementary condition imposed on the potential, the Birman-Schwinger operator is shown to decompose into a singular operator which is of trace class and a finite part with bounded integral kernel. The finite portion belongs to the Hilbert-Schmidt class $\mathcal{S}_2$ and the associated integral kernels valid for any real value $\alpha \in (1,2)$. Finally, we prove that a bound state exists iff it obeys \ref{sec4 : eq4} and is uniquely determined.
\section{Definition and Properties of the Fractional Weyl \\ Operator}
\label{sec1}
Consider an infinite complex and separable Hilbert space $\mathcal{H}=L^2({\mathbb R},dx)$ with inner product defined by
\begin{eqnarray}
<f,g>=\int_{{\mathbb R}}\bar{f}(x)g(x)dx
\label{sec1 : eq1}
\end{eqnarray} 
where $\bar{f}$ denotes the complex conjugate of $f$. 
\begin{definition} 
The nonlocal fractional operator  $\hat{\mathcal{P}}^{\alpha}: \,\, \mathcal{D}(\hat{\mathcal{P}}^{\alpha})\subset \mathcal{H}\rightarrow \mathcal{H}$ defined by \cite{Ref1}
\begin{eqnarray}
(\hat{\mathcal{P}}^{\alpha}\psi)(x)&=& \frac{1}{\cos(\frac{\pi \alpha}{2})}
\left({}_{-\infty}\hat{\mathcal{P}}_x^{\alpha}+{}_{x}\hat{\mathcal{P}}_{\infty}^{\alpha}\right)
\psi(x)=\frac{1}{\cos(\frac{\pi
\alpha}{2})}\left(\hat{\mathcal{P}}_-^{\alpha}+\hat{\mathcal{P}}_+^{\alpha}\right)\psi(x)\non \\
&\stackrel{def}{=}& \!\!\! \! \frac{1}{ \cos(\frac{\pi \alpha}{2})} \frac{1}{\Gamma(m-\alpha)}
\left(\frac{d}{dx}\right)^m \!\!\!
\left( \int_{-\infty}^{x}\frac{\psi(z)}{(x-z)^{\alpha -m +1}}dz +
(-1)^m \int_{x}^{\infty} \!\!\! \frac{\psi(z)}{(z-x)^{\alpha -m +1}}dz \right), \non \\
&=&\frac{1}{ \cos(\frac{\pi \alpha}{2})} \frac{1}{\Gamma(m-\alpha)}
\left(\frac{d}{dx}\right)^m \int_0^{\infty}\frac{\psi(x-u)+(-1)^m \psi(x+u)}{u^{\alpha-m+1}}du
\label{sec1 : eq2}
\end{eqnarray}
is called the Weyl operator.
\label{sec1 : def1} 
\end{definition}
The domain of the Weyl operator is defined by 
\begin{eqnarray}
\mathcal{D}(\hat{\mathcal{P}}^{\alpha}):=\{f, (\hat{\mathcal{P}}^{k}f)\in L^2({\mathbb R}), \, \forall k=1,\cdots , [\alpha]: \,\, \textrm{with}\,\, f\in AC^{[a]}(\Omega)\}
\label{sec1 : eq3}
\end{eqnarray}
where the space $AC^{[\alpha]}(\Omega)$  consists of all functions $f$ which have continous derivatives up to order $[\alpha]-1$ on every compact interval $\Omega \subset {\mathbb R}$ with $f^{(k)}\in AC(\Omega), \, \forall k=1,\cdots,[\alpha]-1$. It is a dense subset of $\mathcal{H}$. \par In \ref{sec1 : eq2}, $\alpha=[\alpha]+\{\alpha\}$ with $[\alpha], \{\alpha\}$ representing the integral and fractional part ($0<\{\alpha\}<1$) of the real number $\alpha$. Also $m=[\alpha]+1$, $\hat{\mathcal{P}}^{\alpha}_{\mp} f$ are 
the left- and right-handed fractional derivatives and $\Gamma$ is the Euler's gamma 
function. From the definition \ref{sec1 : eq2} we note that the left-handed (right-handed) Weyl derivative of a function $\psi$ at a point $x$ depends on all function values to the left (right) of the point. When $\alpha$ is an even integer then the two derivatives are localized and equal while for odd integer values of $\alpha$ both derivatives appear opposite in signs. For $\alpha =2l$ ($l\in {\mathbb N}$) the operator could be defined with a $\frac{1}{2}$-factor to cancel the double contribution of the even order derivatives from the two terms. In this case the Weyl operator coincides to the elliptic differential operator $(-\Delta)^l=(-d^2/dx^2)^l$ of order $2l$. 
\par The Weyl operator is easily checked to be linear since 
\begin{eqnarray}
\left(\hat{\mathcal{P}}^{\alpha}(c_1 f+c_2 g)\right)(x)=c_1(\hat{\mathcal{P}}^{\alpha}f)(x)+c_2 (\hat{\mathcal{P}}^{\alpha} g)(x) \quad \forall f,g\in \mathcal{D}(\hat{\mathcal{P}}^{\alpha}) \quad \textrm{and} \quad c_1,c_2\in {\mathbb C}.
\label{sec1 : eq4}
\end{eqnarray}
\begin{lemma}
The operator $\hat{\mathcal{P}}^{\alpha}$ defined by \ref{sec1 : eq2} is unbounded.
\label{sec1 : lem1} 
\end{lemma}
\begin{proof}
\label{sec1 : pr1} 
\end{proof}
Suppose that $k<\alpha<k+1$ where $k\in {\mathbb Z}_{+}$. The Weyl operator
$\hat{\mathcal{P}}^{\alpha}$ is an extension of 
\begin{eqnarray}
\hat{\mathcal{P}}^{\alpha}_0=\hat{\mathcal{P}}^{\alpha}\Bigr|_{Y}
\label{sec1 : eq5}
\end{eqnarray}
where $Y=\mathcal{D}(\mathcal{P}^{\alpha})\cap L^2[0,1]$. If $\hat{\mathcal{P}}^{\alpha}_0$ is unbounded so is $\hat{\mathcal{P}}^{\alpha}$. Consider the sequence of functions $\{\phi_n\}_{n\in {\mathbb N},\, n\geq 2}$ defined by
\begin{eqnarray}
\phi_n(x)=\left\{\begin{array}{ll} 1-(nx)^{[\alpha]+1}, & x\in [0,\frac{1}{n})\\
0, & x\in [\frac{1}{n},1]. \end{array}\right.
\label{sec1 : eq6}
\end{eqnarray} 
We can verify that $\phi_n\in L^2[0,1]$ since
\begin{eqnarray}
\parallel \phi_n \parallel^2 &=& \int_0^{\frac{1}{n}} \left(1-(nx)^{[\alpha]+1}\right)^2 dx \non \\
&=& \left(1-\frac{2}{[\alpha]+2}+\frac{1}{2([\alpha]+1)+1}\right)\frac{1}{n}.
\label{sec1 : eq7}
\end{eqnarray}
Denoting by  
\begin{eqnarray}
g_n(x)&=& \int_0^{\frac{1}{n}}\frac{\phi_n(x-u)+(-1)^m \phi_n(x+u)}{u^{\alpha-[\alpha]}}du 
= \frac{\left(1+(-1)^{[\alpha]+1}\right)}{[\alpha]+1-\alpha}\frac{1}{n^{[\alpha]+1-\alpha}}\non \\
&+& \!\!\! (-1)^{[\alpha]+2}\sum_{k=0}^{[\alpha]+1}\!\! \left(\!\! \begin{array}{c} [\alpha]+1 \\ k \end{array}\!\! \right) \left(1+(-1)^k \right)\frac{1}{2([\alpha]+1)-(\alpha+k)}\frac{1}{n^{[\alpha]+1-(\alpha+k)}}x^k
\label{sec1 : eq8}
\end{eqnarray}
and substituting \ref{sec1 : eq8} into \ref{sec1 : eq2} a straightforward calculation gives
\begin{eqnarray}
\left(\hat{\mathcal{P}}_0^{\alpha}\phi_n \right)(x)= \left\{\begin{array}{ll} -\frac{1}{ \cos(\frac{\pi \alpha}{2})} \left(1+(-1)^{[\alpha]+1}\right)\frac{1}{[\alpha]+1-\alpha}\frac{\Gamma([\alpha]+2)}{\Gamma([\alpha]+1-\alpha)}n^{[\alpha]+\alpha+1}, & 0<x<\frac{1}{n}\\ 0, & \frac{1}{n}\leq x<1. \end{array}\right.
\label{sec1 : eq9}
\end{eqnarray}
The quotient
\begin{eqnarray}
\frac{\parallel \hat{\mathcal{P}}_0^{\alpha}\phi_n \parallel }{\parallel \phi_n\parallel }&=&\frac{1}{ \left|\cos(\frac{\pi \alpha}{2})\right|} \left(1+(-1)^{[\alpha]+1}\right) \frac{\Gamma(2+[\alpha])}{\Gamma(2+[\alpha]-\alpha)} \non \\
&\times& \frac{1}{\left(1-\frac{2}{[\alpha]+2}+\frac{1}{2([\alpha]+1)+1}\right)^{\frac{1}{2}}} n^{\alpha+[\alpha]+1}>n^{\alpha+[\alpha]+1}
\label{sec1 : eq10}
\end{eqnarray}
shows that $\hat{\mathcal{P}}_0^{\alpha}$ is unbounded for $[\alpha]=$odd. If $[\alpha]=$even then the sequence \ref{sec1 : eq6} is defined by $\phi_n(x)=1-(nx)^{[\alpha]}$ for $x\in [0,\frac{1}{n})$, and zero elsewhere. In this case we derive a similar expression to \ref{sec1 : eq10} with the only difference the substitution $[\alpha]\rightarrow [\alpha]+1$.
\begin{lemma}
The linear operator $\hat{\mathcal{P}}^{\alpha}: \, \mathcal{D}(\hat{\mathcal{P}}^{\alpha})\rightarrow \mathcal{H}$ defined by \ref{sec1 : eq2} is symmetric.
\label{sec1 : lem2} 
\end{lemma}
\begin{proof}
\label{sec1 : pr2} 
\end{proof}
We first prove that if the functions $f, f^{(1)}\in L^2({\mathbb R_+})$ with $f$ absolute continuous then 
\begin{eqnarray}
\lim_{x\rightarrow \infty}f(x)=0.
\label{sec1 : eq10a} 
\end{eqnarray}
Integrating $\frac{d}{dy}|f(y)|^2$, we get
\begin{eqnarray}
|f(x)|^2=|f(0)|^2+\int_0^x\left(f^{(1)}(y)\overline{f(y)}+\overline{f^{(1)}(y)}f(y)\right)dy.
\label{sec1 : eq10b} 
\end{eqnarray}
The integral on the right-hand side converges as $x\rightarrow \infty$ since $f,f^{(1)}\in L^2({\mathbb R_+})$. Therefore the limit \ref{sec1 : eq10a} exists and can be only zero otherwise $f\notin L^2$.
\par This argument can be generalized in our case and therefore the following boundary conditions hold 
\begin{eqnarray}
\lim_{|x|\rightarrow \infty}f^{(k)}(x)=0=\lim_{|x|\rightarrow \infty}g^{(k)}(x), \,\, k=0,\cdots, [\alpha]-1. 
\label{sec1 : eq11}
\end{eqnarray}
By definition, a linear and densely defined operator $\hat{\mathcal{A}}$ on a Hlibert space $\mathcal{H}$ is symmetric if
\begin{eqnarray}
<\hat{\mathcal{A}}f,g>=<f,\hat{\mathcal{A}}g>\quad \forall f, g\in \mathcal{D}(\hat{\mathcal{A}}).
\label{sec1 : eq11a} 
\end{eqnarray}
This can be established for $\hat{\mathcal{P}}^{\alpha}$ by proving the relation 
\begin{eqnarray}
<{}_{-\infty}\hat{\mathcal{P}}_x^{\alpha}f,g>=<f,{}_{x}\hat{\mathcal{P}}_{\infty}^{\alpha}g>, \quad \forall f, g\in \mathcal{D}(\hat{\mathcal{P}}^{\alpha}).
\label{sec1 : eq12}
\end{eqnarray}
Using partial integration and \ref{sec1 : eq11}, one has 
\begin{eqnarray}
\int_{{\mathbb R}} \left({}_{-\infty}\hat{\mathcal{P}}_x^{\alpha}\bar{f}\right) (x)g(x)dx
&=& \frac{1}{\Gamma([\alpha]+1-\alpha)}\int_{{\mathbb R}} \left(\frac{d}{dx} \right)^{[\alpha]+1} \left(\int_{0}^{\infty} \frac{\bar{f}(x-u)}{u^{\alpha-[\alpha]}}du\right)g(x)dx \nonumber \\
&=& \frac{(-1)^{[\alpha]+1}}{\Gamma([\alpha]+1-\alpha)}\int_{{\mathbb R}} \left(\int_{0}^{\infty} \frac{\bar{f}(x-u)}{u^{\alpha-[\alpha]}}du\right)\left(\frac{d}{dx} \right)^{[\alpha]+1}g(x)dx \nonumber \\
&=& \frac{(-1)^{[\alpha]+1}}{\Gamma([\alpha]+1-\alpha)}\int_{{\mathbb R}} \bar{f}(s)\left(\frac{d}{ds} \right)^{[\alpha]+1} \left(\int_{s}^{\infty}\frac{g(x)}{(x-s)^{\alpha-[\alpha]}}dx\right)ds \nonumber \\
&=&\int_{{\mathbb R}} \bar{f}(x)\, \left({}_{x}\hat{\mathcal{P}}_{\infty}^{\alpha}g\right)(x)dx. 
\label{sec1 : eq13}
\end{eqnarray}
\noi From this property we conclude that the Hilbert-adjoint operator $(\hat{\mathcal{P}}^{\alpha})^*$ is an extension of $\hat{\mathcal{P}}^{\alpha}$ thus
\begin{eqnarray}
\mathcal{D}(\hat{\mathcal{P}}^{\alpha})\subset \mathcal{D}((\hat{\mathcal{P}}^{\alpha})^*) \quad \rm{and} \quad \hat{\mathcal{P}}^{\alpha}=(\hat{\mathcal{P}}^{\alpha})^*\biggr|_{\mathcal{D}(\hat{\mathcal{P}}^{\alpha})}.
\label{sec1 : eq14a} 
\end{eqnarray}
 
\begin{lemma}
The operator $\hat{\mathcal{P}}^{\alpha}$ defined by \ref{sec1 : eq2} is self-adjoint. 
\label{sec1 : lem4} 
\end{lemma}
\begin{proof}
\label{sec1 : pr4} 
\end{proof}
For $f\in S({\mathbb R})$ \footnote{$S({\mathbb R})$ is the set of infinitely differentiable and rapidly decreasing functions on ${\mathbb R}$, namely $sup_{x\in {\mathbb R}}|x^n (D^m f)(x)|<\infty, \forall m,n=0,1,\cdots$. This space is usually called the Schwartz space.} the Fourier transform of $\hat{\mathcal{P}}^{\alpha}f$ is (see Appendix A for the proof)
\begin{eqnarray}
\left(\mathcal{F}(\hat{\mathcal{P}}^{\alpha}f)\right)(p)&=& \frac{1}{2\pi}\int_{-\infty}^{\infty}e^{-ipx} (\hat{\mathcal{P}}^{\alpha}f)(x) dx= \frac{1}{2\pi}\int_{-\infty}^{\infty} (\hat{\mathcal{P}}^{\alpha}e^{-ipx}) f(x) dx \non \\
&=&2|p|^{\alpha}(\mathcal{F}f)(p), \quad \textrm{for} \quad p\in {\mathbb R}.
\label{sec1 : eq19}
\end{eqnarray}
The Fourier transform $\mathcal{F}: \, S({\mathbb R})\rightarrow S({\mathbb R})$ is bijective and since $S({\mathbb R})$ is a dense subspace of $L^2({\mathbb R})$ it can be extended continuously to $L^2({\mathbb R})$ in a unique manner \cite{Ref4a}. This extension will still be denoted by $\mathcal{F}$ and satisfies Plancherel's equation 
\begin{eqnarray}
<\mathcal{F}f,\mathcal{F}g>_{L^2}=<f,g>_{L^2}, \quad \textrm{for all} \quad f,g\in L^2({\mathbb R}).
\label{sec1 : eq19a}
\end{eqnarray}
Thus $\mathcal{F}$ is a unitary operator and
\begin{eqnarray}
\hat{\mathcal{P}}^{\alpha}f=2\mathcal{F}^{-1}(|p|^{\alpha}\mathcal{F}f).
\label{sec1 : eq18b}
\end{eqnarray}
The multiplication operator $M_{2|p|^{\alpha}}$ defined by 
\begin{eqnarray}
M_{2|p|^{\alpha}}f &:=& 2|p|^{\alpha}\cdot f \quad \textrm{for all $f$ in the domain}  \non \\ \mathcal{D}(M_{2|p|^{\alpha}}) &:=& \{f \in L^2({\mathbb R}): \, 2|p|^{\alpha}\cdot f \in L^2({\mathbb R}) \} 
\label{sec1 : eq18c}
\end{eqnarray}
is unitarily equivalent to $\hat{\mathcal{P}}^{\alpha}$ according to \ref{sec1 : eq18b}. Since $M_{|p|^{\alpha}}$ is self-adjoint we conclude that $\hat{\mathcal{P}}^{\alpha}$ is also self-adjoint \footnote{An alternative way to prove self-adjointness would be to show that $( \hat{\mathcal{P}}^{\alpha})^* \subset \hat{\mathcal{P}}^{\alpha}$ which combined with Lemma \ref{sec1 : lem2} results in $\hat{\mathcal{P}}^{\alpha}=(\hat{\mathcal{P}}^{\alpha})^*$.}. As a result $\hat{\mathcal{P}}^{\alpha}=(\hat{\mathcal{P}}^{\alpha})^*$ and the Hilbert-adjoint $(\hat{\mathcal{P}}^{\alpha})^*$ is a closed operator. Thus $\hat{\mathcal{P}}^{\alpha}$ is also closed. 

\section{The free Green's function $G_{\alpha}(x,y;\kappa)$ and its derivative $\partial G_{\alpha}(x,y;\kappa)/\partial \kappa$}
\label{sec2}

The FGF $G_{\alpha}(x,y;\kappa)$ regarded as a distribution with respect to x and considering y as parameter, obeys \footnote{The case of $\kappa=0$ produces the Green's function 
\begin{displaymath}G_{\alpha}(x-y;\kappa=0)=\frac{1}{\pi K_{\alpha}|x-y|}\Gamma(1-\alpha)\cos\left(\frac{\pi}{2}(1-\alpha)\right), \quad 0<\alpha<1. \end{displaymath}}
\begin{eqnarray}
&& \left(K_{\alpha}\hat{\mathcal{P}_x}^{\alpha}+\kappa^2 \right)G_{\alpha}(x,y;\kappa)=\delta(x-y), \quad E_0=\kappa^2\neq 0 \non \\
&\textrm{with}& \quad \lim_{|x|\rightarrow \infty}G_{\alpha}(x,y;\kappa)=0.
\label{sec2 : eq6}
\end{eqnarray} 
To solve \ref{sec2 : eq6} we Fourier transformed it 
\begin{eqnarray}
\int_{-\infty}^{+\infty}e^{-ipx} \left[\left(K_{\alpha}|p|^{\alpha}+\kappa^2\right)\tilde{G}_{\alpha}(p,y;\kappa)-e^{ipy}\right]dp=0,
\label{sec2 : eq7}
\end{eqnarray}
and the solution is 
\begin{eqnarray}
\tilde{G}_{\alpha}(p,y;\kappa)=\frac{e^{ipy}}{K_{\alpha}|p|^{\alpha}+\kappa^2}+f(|p|)\delta(p).
\label{sec2 : eq8}
\end{eqnarray}
The function $f$ is an arbitrary polynomial of $|p|$ with positive lowest order. The second term in \ref{sec2 : eq8} does not contribute to $G_{\alpha}(x,y;\kappa)$ since $\int_{-\infty}^{\infty}e^{-ip(x-y)}f(|p|)\delta (p)=f(0)=0$. Also since $\frac{1}{K_{\alpha}|p|^{\alpha}+\kappa^2}\in L^2({\mathbb R})\cap L_1({\mathbb R})$ the function $G_{\alpha}(x-y;\kappa)$ is continuous and belongs to $L^2({\mathbb R})$.  The FGF is given by the Fourier-Laplace transformation
\begin{eqnarray}
G_{\alpha}(x-y;\kappa)&=&\frac{1}{2\pi}\int_{-\infty}^{+\infty}e^{-ipx}\tilde{G}(p,y;\kappa) dp=\frac{1}{2\pi}\int_{-\infty}^{+\infty} \frac{e^{-ip(x-y)}}{K_{\alpha}|p|^{\alpha}+\kappa^2}dp \non \\
&=& \frac{1}{2\pi}\int_{-\infty}^{+\infty} e^{-ip(x-y)}\left(\int_{0}^{\infty} e^{-s(K_{\alpha}|p|^{\alpha}+\kappa^2)} ds \right) dp \non \\
&=& \frac{1}{2\pi (K_{\alpha})^{\frac{1}{\alpha}}}\int_{0}^{\infty} e^{-s\kappa^2} s^{-\frac{1}{\alpha}}\left(\int_{-\infty}^{+\infty} e^{-|u|^{\alpha}+i(sK_{\alpha})^{-\frac{1}{\alpha}}(x-y)u}du \right) ds.
\label{sec2 : eq9}
\end{eqnarray} 
In the derivation of \ref{sec2 : eq9} it is eligible to interchange the order of integration by applying Fubini's theorem since the integrand is continous, thus measurable, and the integral is absolutely convergent 
\begin{eqnarray}
\int_0^{\infty}\!\!\left(\int_{-\infty}^{\infty}\left|e^{-ip(x-y)}e^{-s(K_{\alpha}p^{\alpha}+\kappa^2)}\right|ds\right) dp &\leq& \frac{2\kappa^{2\left(\frac{1}{\alpha}-1\right)}}{\alpha (K_{\alpha})^{\frac{1}{\alpha}}} B\left(\frac{1}{\alpha},\frac{\alpha -1}{\alpha} \right)<\infty, \,\, 1<\alpha\leq 2.
\label{sec2 : eq10}
\end{eqnarray}
The condition $1<\alpha \leq 2$ is implied by the integral representation of the beta function.
\par In the last equality of \ref{sec2 : eq9} the integral in the parenthesis represents the Fourier transform of the Fox's H-function with well-known properties \cite{Ref5,Ref6}. We insert into \ref{sec2 : eq9} the contour integral representation of $e^{-u^{\alpha}}$ obtained by expressing it as the inverse Mellin transform of $\Gamma\left(\frac{s}{\alpha}\right)$ (see Appendix B for the proof)
\begin{eqnarray}
e^{-u^{\alpha}}=\frac{1}{2\pi i \alpha}\int_{c-i\infty}^{c+i\infty}u^{-s}\Gamma \left(\frac{s}{\alpha}\right)ds, \quad c>0
\label{sec2 : eq11}
\end{eqnarray} 
and obtain
\begin{eqnarray}
G_{\alpha}(x-y;\kappa)&=&\frac{1}{2\alpha \pi^2 i}\frac{1}{\left(K_{\alpha}\right)^{\frac{1}{\alpha}}\kappa^{2\left(1-\frac{1}{\alpha}\right)}} \int_{c-i\infty}^{c+i\infty}\Gamma(1-s)\Gamma\left(\frac{s}{\alpha}\right)\Gamma\left(1-\frac{s}{\alpha}\right)\sin \left(\frac{s\pi}{2}\right)z^{s-1}ds, \non \\
&& \,\, \textrm{where} \,\, 0<c<1 \,\, \textrm{and} \,\, z=\left(\frac{|x-y|^{\alpha}\kappa^2}{K_{\alpha}}\right)^{\frac{1}{\alpha}}. 
\label{sec2 : eq12}
\end{eqnarray}
In writing \ref{sec2 : eq12} we used the integral formula [3.761.9] of \cite{Ref9} 
\begin{eqnarray}
\int_0^{\infty}x^{\mu-1} \cos(zx)dx=\frac{\Gamma(\mu)}{z^{\mu}}\cos \left(\frac{\mu \pi}{2}\right), \,\, z>0, \, 0<\textrm{Re} \mu<1,
\label{sec2 : eq11a}
\end{eqnarray} 
and the definition of the gamma function. 
\par We proceed by investigating the asymptotic behaviour of \ref{sec2 : eq12} using Cauchy's residue theorem on a rectangular contour $\mathcal{L}$ described counterclockwised (positive direction) on the complex s-plane. 
\begin{figure}[!ht]
\centering
\includegraphics[147,489][542,779]{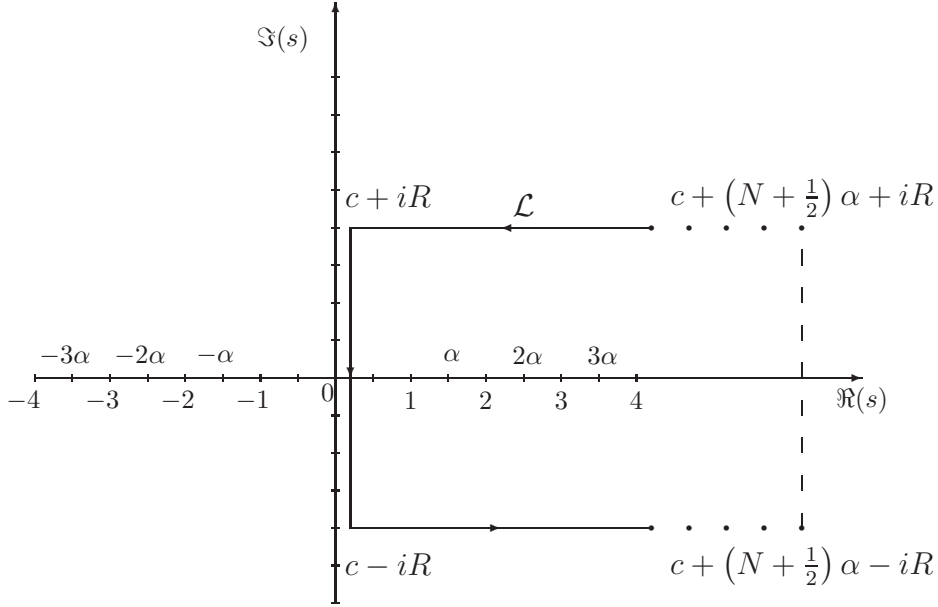}
\caption{The rectangular contour $\mathcal{L}$ on the complex $s$-plane. The sum of the residues of the integrand at the simple poles located at the positive real axis provide the asymptotic behaviour of the free Green's function in the neighbourhood of the point $z=0$.}
\label{sec2 : fig1}
\end{figure}
We study first the case when index $\alpha$ is irrational. The integrand 
\begin{eqnarray}
f(s)=\Gamma(1-s)\Gamma\left(\frac{s}{\alpha}\right)\Gamma\left(1-\frac{s}{\alpha}\right)\sin \left(\frac{s\pi}{2}\right)z^{s-1}
\label{sec2 : eq13}
\end{eqnarray}
in the positive real axis has simple poles given by the union of the simple poles of gamma functions $\Gamma(1-s)$ and $\Gamma(1-\frac{s}{\alpha})$ \footnote{In principle one should have included the poles of $\Gamma(\frac{s}{\alpha})$ but these have already been encountered in those of $\Gamma(1-\frac{s}{\alpha})$ excluding the zero pole. The reason is the cancelation by the zero of sine function.}
\begin{eqnarray}
\mathcal{P}_+ = \{s_n=1+n, \,\, n=0,1,\cdots \}\cup \{s_m=m\alpha, \,\, m=1,2, \cdots\}.
\label{sec2 : eq14}
\end{eqnarray}
Special care should be paid on those poles at which the sine function vanishes. The series representation of the gamma function $\Gamma(z)$ near the poles $z=0,-1,-2,\cdots$ is
\begin{eqnarray}
\Gamma(z)\simeq\frac{(-1)^n}{n! (z+n)}+\frac{(-1)^n \psi(n+1)}{n!}+O(z+n), \quad z\rightarrow -n, \,\, n\in {\mathbb N}\cup \{0\}.
\label{sec2 : eq15}
\end{eqnarray}
where $\psi(x)=\frac{d \ln \Gamma(x)}{dx}$ and for $n\in {\mathbb N}$ it takes the values 
\begin{eqnarray}
\psi(n+1)=-\gamma+\sum^{n}_{p=1}\frac{1}{p}.
\label{sec2 : eq16}
\end{eqnarray}
In \ref{sec2 : eq16} $\gamma$ is Euler's constant. Thus, combining \ref{sec2 : eq15} with the Taylor expansion of sine function around a given zero we can determine whether the pole is cancelled by the zero or not. For the even poles $s_k=2k, \,\, k\in {\mathbb N}$ we have
\begin{eqnarray}
\Gamma(1-s)\sin\left(\frac{s\pi}{2}\right) &=& \Gamma(1-s)\cos\left(\frac{\pi}{2}(1-s)\right) \non \\ 
&\simeq& \frac{(-1)^k}{(2k-1)!}\frac{\pi}{2}+O(2k-s), \,\, 1-s\rightarrow -(2k-1). 
\label{sec2 : eq17}
\end{eqnarray} 
Therefore the even poles of $\Gamma(1-s)$ have all been cancelled out by the zeroes of $\sin(s\pi/2)$. The set \ref{sec2 : eq14} is finally reduced to
\begin{eqnarray}
\mathcal{P}_{+,red.}= \{s_n=2n+1, \,\, n=0,1,\cdots \}\cup \{s_m=m\alpha, \,\, m=1,2, \cdots\}.
\label{sec2 : eq18}
\end{eqnarray}
If we close the contour $\mathcal{L}$ on the right-hand side then the asymptotic behaviour of $G(x,y;\kappa)$ near the point $z=0$ (or the small distance behaviour) is given by
\begin{eqnarray}
G_{\alpha}^R(x-y;\kappa)&=& \frac{1}{2\alpha \pi^2 i}\frac{1}{\left(K_{\alpha}\right)^{\frac{1}{\alpha}}\kappa^{2\left(1-\frac{1}{\alpha}\right)}} (2\pi i)\sum_i \textrm{Res}(f(s_i), s_i\in \mathcal{P}_{+,red.}) \non \\
&=&\frac{1}{\alpha \left(K_{\alpha}\right)^{\frac{1}{\alpha}}}\frac{1}{\kappa^{2\left(1-\frac{1}{\alpha}\right)}} \Biggl[ \sum_{m=0}^{\infty}(-1)^m \frac{1}{\Gamma(2m+1)\sin\left(\frac{(2m+1)\pi}{\alpha}\right)} z^{2m} \non \\
&+& \frac{\alpha}{2} \sum_{m=1}^{\infty} \frac{(-1)^{m+1}}{\Gamma(\alpha m) \cos\left(\frac{m\alpha \pi }{2}\right)} z^{\alpha m -1} \Biggr].
\label{sec2 : eq19}
\end{eqnarray}
The second series was simplified by applying the identity $\Gamma(z)\Gamma(-z)=-\frac{\pi}{z\sin (\pi z)}$. Using the ratio test we can check that the two series absolutely converge. Indeed
\begin{eqnarray}
\rho_1&=& \lim_{m\rightarrow \infty} \left|\frac{\Gamma(2m+1)\sin\left(\frac{(2m+1)\pi}{\alpha}\right)}{\Gamma(2m+3)\sin\left(\frac{(2m+3)\pi}{\alpha}\right)}\right|\leq \lim_{m\rightarrow \infty}\left|\frac{1}{(2m+1)(2m+2)\sin\left(\frac{(2m+3)\pi}{\alpha}\right)}\right|=0 \non \\
\rho_2 &\leq& \lim_{m\rightarrow \infty} 
\left|\frac{\Gamma(\alpha m)}{\Gamma(\alpha (m+1)) \cos\left(\frac{(m+1)\alpha \pi }{2}\right)} \right|=0. 
\label{sec2 : eq20}
\end{eqnarray}
The series absolutely converge for every value of $z$ in the interval $(-R,R)=(-\infty,\infty)$. Note that the asymptotic expansion of Green's function around the point $z=0$ consists of two power series in $z$, one with positive integer and another with non-negative real powers. Also, \ref{sec2 : eq19} in the $\kappa \rightarrow 0$ limit decomposes into a singular part, a constant $\kappa$-independent term and a regular vanishing part according to
\begin{eqnarray}
G_{\alpha}^R(x-y;\kappa) &=& \frac{1}{\alpha \left(K_{\alpha}\right)^{\frac{1}{\alpha}}}\frac{1}{\kappa^{2\left(1-\frac{1}{\alpha}\right)}} \frac{1}{\sin\left(\frac{\pi}{\alpha}\right)}+\frac{1}{2 K_{\alpha}\Gamma(\alpha) \cos\left(\frac{\alpha \pi}{2}\right)}|x-y|^{\alpha-1} \non \\
&+& \frac{1}{\alpha \left(K_{\alpha}\right)^{\frac{1}{\alpha}}}\frac{1}{\kappa^{2\left(1-\frac{1}{\alpha}\right)}} \Biggl[ \sum_{m=1}^{\infty}(-1)^m \frac{1}{\Gamma(2m+1)\sin\left(\frac{(2m+1)\pi}{\alpha}\right)} z^{2m} \non \\
&+& \frac{\alpha}{2} \sum_{m=2}^{\infty} \frac{(-1)^{m+1}}{\Gamma(\alpha m) \cos\left(\frac{m\alpha \pi }{2}\right)} z^{\alpha m -1} \Biggr].
\label{sec2 : eq20a}
\end{eqnarray}
Expression \ref{sec2 : eq20a} makes sense only for irrational values of $\alpha$ since otherwise the denominators of the sums become infinite for specific values of $m$. 
\par The asymptotic behaviour of $G_{\alpha}(x,y;\kappa)$ in the neighbourhood of the poit $z=\infty$ (or the long distance behaviour) is derived by closing the contour on the left-hand side  
\begin{eqnarray}
G_{\alpha}^L(x-y;\kappa)= \frac{1}{\left(K_{\alpha}\right)^{\frac{1}{\alpha}}}\frac{1}{\kappa^{2\left(1-\frac{1}{\alpha}\right)}}\frac{\alpha}{\pi} \sum_{m=1}^{\infty}(-1)^{m+1} \Gamma(1+m\alpha) \sin \left(\frac{m\alpha \pi}{2}\right)  z^{-(m\alpha+1)}.
\label{sec2 : eq21}
\end{eqnarray}
Each term of the series \ref{sec2 : eq21}, for fixed m, in the $|x|\rightarrow \infty$ limit vanishes.  Even for big values of m, using Stirling's asymptotic formula \ref{apA : eq3}, one can show that for $m \leq |x|$ the corresponding term still vanishes. Thus the Green's function fulfills the boundary condition. Actually, the boundary condition is implied by the Riemann-Lebesque theorem according to which if $u\in L_1$ then $(\mathcal{F}u)(y) \rightarrow 0$ in ${\mathbb C}$ when $|y|\rightarrow \infty$.  
\par For rational values of $\alpha=\frac{p}{q}$ with $q<p<2q$ and $p,q\in {\mathbb N}$, the calculation becomes cumbersome since there are poles from different gamma's which are also zeroes of sine function and the result is $\alpha$-sensitive \footnote{What we mean by this is the following: it is impossible to find a closed expression that valids for all rational values of $\alpha$.}. 
\par We can check our computation for $\alpha=2$ starting from \ref{sec2 : eq18}. The residues of $f(s)$ at $s_k=2k, \, \, k\in {\mathbb N}$ are \footnote{The same result can be verified directly from \ref{sec2 : eq19} by taking the $\alpha \uparrow 2$ limit 
\begin{displaymath}
\lim_{\alpha \uparrow 2}G_{\alpha}(x,y;\kappa)=G_2(x,y;\kappa).
\end{displaymath}
}
\begin{eqnarray}
\textrm{Res}(f(s_k),s_k=2k)=\frac{(-1)^k}{\Gamma(2k)}\frac{\pi}{2}\Gamma(k)\frac{(-1)^{k-1}}{\Gamma(k)}z^{2k-1}=-\frac{\pi}{2}\frac{1}{\Gamma(2k)}z^{2k-1}.
\label{sec2 : eq22}
\end{eqnarray}
The total contribution is then
\begin{eqnarray}
G_{\alpha=2}^R(x-y,\kappa)=\frac{1}{2\sqrt{K_2}}\frac{1}{\kappa} e^{-\frac{|x-y|\kappa}{\sqrt{K_2}}}
\label{sec2 : eq23}
\end{eqnarray}
which agrees with the expected result in one dimension. The contribution of \ref{sec2 : eq21} from the left sector vanishes due to the sine function and this implies that the boundary condition is satisfied automatically.
\par Next we examine the asymptotic behaviour of $\partial G_{\alpha}(x,y;\kappa)/\partial \kappa$ near the point $z=0$. Following similar steps as to the case of the FGF we obtain 
\begin{eqnarray}
\frac{\partial G_{\alpha}(x-y;\kappa)}{\partial \kappa}&=&-\frac{1}{\alpha \pi^2 i} \frac{1}{\left(K_{\alpha}\right)^{\frac{1}{\alpha}}\kappa^{3-\frac{2}{\alpha}}} \int_{c-i\infty}^{c+i\infty}\Gamma(1-s)\Gamma\left(\frac{s}{\alpha}\right)\Gamma\left(2-\frac{s}{\alpha}\right)\sin \left(\frac{s\pi}{2}\right)z^{s-1}ds, \non \\
&& \,\, \textrm{where} \,\, 0<c< 1 \,\, \textrm{and} \,\, z=\left(\frac{|x-y|^{\alpha}\kappa^2}{K_{\alpha}}\right)^{\frac{1}{\alpha}}.
\label{sec2 : eq24}
\end{eqnarray}
The reduced set of simple poles in the positive real axis $\Re(s)$ is then 
\begin{eqnarray}
\mathcal{P}_{+,red.}= \{s_n=2n+1, \,\, n=0,1,\cdots \}\cup \{s_m=m\alpha, \,\, m=2, \cdots\}.
\label{sec2 : eq25}
\end{eqnarray}
Note that the pole $s_1=\alpha$ is absent in this case while in the negative real axis $s_{-1}=-\alpha$ is present. Applying Cauchy's residue theorem using the rectangular contour $\mathcal{L}$ on the right-hand side we find
\begin{eqnarray}
\frac{\partial G_{\alpha}^R(x-y;\kappa)}{\partial \kappa}&=& -\frac{2}{\alpha \left(K_{\alpha}\right)^{\frac{1}{\alpha}} \kappa^{3-\frac{2}{\alpha}}}\Biggl[ \sum_{m=0}^{\infty} \frac{(-1)^m}{\Gamma(2m+1)\sin\left(\frac{(2m+1)\pi}{\alpha}\right)}\left(1-\frac{2m+1}{\alpha}\right) z^{2m} \non \\
&+& \frac{\alpha}{2} \sum_{m=2}^{\infty} \frac{(-1)^{m+1}}{\Gamma(\alpha m) \cos\left(\frac{m\alpha \pi }{2}\right)}(1-m) z^{\alpha m -1} \Biggr].
\label{sec2 : eq26}
\end{eqnarray} 
The singular term is
\begin{eqnarray}
\frac{\partial G_{\alpha, \textrm{sing.}}^R(x-y;\kappa)}{\partial \kappa}=-\frac{2}{\alpha \left(K_{\alpha}\right)^{\frac{1}{\alpha}} \kappa^{3-\frac{2}{\alpha}}}\frac{1}{\sin\left(\frac{\pi}{\alpha}\right)}\left(1-\frac{1}{\alpha}\right)
\label{sec2 : eq27}
\end{eqnarray}
and all other terms vanish in the $\kappa \downarrow 0$ limit. 
\section{Representation of the Birman-Schwinger operator for the Schr\"{o}dinger operator}
\label{sec3}

In the present work we consider the Schr\"{o}dinger operator 
\begin{eqnarray}
\hat{H}=K_{\alpha}\hat{\mathcal{P}}^{\alpha}-g |\hat{V}|
\label{sec3 : eq01}
\end{eqnarray}
in one dimension acting on the space $L^2({\mathbb R})$. The constant $K_{\alpha}$ has dimensions $[K_{\alpha}]=\frac{[M][L]^{\alpha+2}}{[T]^2}$, the nonlocal fractional operator $\hat{\mathcal{P}}^{\alpha}$ is defined by \ref{sec1 : eq2} and the `coupling constant' $g$ is real and takes values in the interval $(0,1]$. 
\par $\hat{V}$ is the multiplication operator with action and domain defined by
\begin{eqnarray}
\hat{V}f &:=& V\cdot f, \quad f\in \mathcal{D}(\hat{V}) \non \\
\mathcal{D}(\hat{V}) &:=& \{f\in L^2({\mathbb R}): \, V\cdot f\in L^2({\mathbb R})\}.
\label{sec3 : eq02}
\end{eqnarray}
This operator is usually referred as the potential. In the sequel we shall assume that the potential $V(x)$ satisfies the conditions 
\begin{align}
V(\cdot)  & \in C({\mathbb R}), \non \\
V(x)<0,   & \,\, \forall x\in {\mathbb R}, \non \\
\lim_{|x|\rightarrow \infty}V(x) &= 0.
\label{sec3 : eq1}
\end{align}
The operator $\hat{H}$ is self-adjoint, acts on the Hilbert space $\mathcal{H}$, with domain  $\mathcal{D}(\hat{H})=\mathcal{D}(\mathcal{\hat{\mathcal{P}}^{\alpha}})\subseteq \mathcal{D}(\hat{V})$ densely defined in $\mathcal{H}$.
\par The eigenvalue problem for bound states is  
\begin{eqnarray}
\left(K_{\alpha}\hat{\mathcal{P}}^{\alpha}-g |V(x)| \right)\psi(x)=-E\psi(x), \quad g\in(0,1], \,\, E\in {\mathbb R}_+/\{0\}, \,\, V\in{\mathbb R}.
\label{sec3 : eq03}
\end{eqnarray}
The solution of the inhomogeneous equation \ref{sec3 : eq03} is given by
\begin{eqnarray}
\psi(x)=g \int G(x,y;\kappa)|V(y)|\psi(y)dy.
\label{sec3 : eq2}
\end{eqnarray}
An equivalent way to determine the solution would be the Birman-Schwinger transformation according to which if $\psi\in \mathcal{D}(\hat{H})\subset L^2({\mathbb R})$ solves \ref{sec3 : eq03}, then  $\phi(x)=\sqrt{|V(x)|}\psi(x)$ solves equation \cite{Ref2,Ref3,Ref4}
\begin{eqnarray}
\left(I-g \hat{D}_{\alpha}\right)\phi(x)=0.
\label{sec3 : eq3}
\end{eqnarray}
Equation \ref{sec3 : eq3} can be casted into the form
\begin{eqnarray}
\frac{1}{g}\hat{I}\phi(x)=\left(\hat{D}_{\alpha}\phi\right)(x)=\int D_{\alpha}(x,y;\kappa)\phi(y) dy =-\int \sqrt{|V(x)|} G_{\alpha}(x,y;\kappa) V^{\frac{1}{2}}(y) \phi(y) dy
\label{sec3 : eq4}
\end{eqnarray} 
where $\hat{I}$ is the identity operator in $\mathcal{D}(\hat{H})$. The function $D_{\alpha}(x,y;\kappa)$ in \ref{sec3 : eq4} is the integral kernel of the Birman-Schwinger operator
\begin{eqnarray}
\hat{D}_{\alpha}=-\sqrt{|\hat{V}|}\left(K_{\alpha}\hat{\mathcal{P}}^{\alpha}+E\right)^{-1}  \hat{V}^{\frac{1}{2}}, 
\label{sec3 : eq5}
\end{eqnarray}  
with $V^{\frac{1}{2}}(x)=\sqrt{|V(x)|} \, \textrm{sign}(V)=-\sqrt{|V(x)|}$. 
\par In what follows we shall need the following statement \cite{Ref5a}. 
\begin{proposition}
The Birman-Schwinger operator $\hat{D}_{\alpha}$ is compact for $E\in {\mathbb R}_+/\{0\}$.
\label{sec3 : prop1} 
\end{proposition}
\par Every compact operator can be decomposed into a finite rank part and an operator which has small norm. The finite part can be realized as a small perturbation on the singular part, in the $\kappa\downarrow 0$ limit, which comprise the dominant contribution. So we have the following lemma:
\begin{lemma}
\label{sec3 : lem1}
If the potential satisfies the additional condition:
\begin{eqnarray}
\int (1+|x|)^{2(\alpha-1)} |V(x)| dx &<& \infty 
\label{sec3 : eq6a} 
\end{eqnarray}
then the Birman-Schwinger operator $\hat{D}_{\alpha}$ defined by \ref{sec3 : eq5} is of trace class and has the representation
\begin{eqnarray}
\hat{D}_{\alpha}=\hat{D}_{\alpha, \textrm{sing.}}+\hat{D}_{\alpha, \textrm{fin.}}
\label{sec3 : eq7}
\end{eqnarray}
where: 
\begin{description}
\item[($i$)] The operator $\hat{D}_{\alpha, \textrm{sing.}}$ has integral kernel
\begin{eqnarray}
D_{\alpha, \textrm{sing.}}(x,y;\kappa)=-\frac{1}{\alpha (K_{\alpha})^{\frac{1}{\alpha}}\kappa^{2\left(1-\frac{1}{\alpha}\right)}\sin\left(\frac{\pi}{\alpha}\right)}\sqrt{|V(x)|} V(y)^{\frac{1}{2}}.
\label{sec3 : eq8}
\end{eqnarray} 
\item[($ii$)] The operator $\hat{D}_{\alpha, \textrm{fin.}}$ has integral kernel with the estimate
\begin{eqnarray}
\left|D_{\alpha, \textrm{fin.}}(x,y;\kappa)\right| \leq M(\alpha) \sqrt{|V(x)|}|x-y|^{\alpha-1}\sqrt{|V(y)|},
\label{sec3 : eq9}
\end{eqnarray}
where $M(\alpha)\geq M(2)$ and does not depend on $x$ and $\kappa$. Also it belongs to the Hilbert-Schmidt class $\mathcal{S}_2$ and its Hilbert-Schmidt norm is uniformly bounded. 
\end{description}
\end{lemma}
\begin{proof}\non
\end{proof}
The operator $\hat{D}_{\alpha}$ is of trace class since
\begin{eqnarray}
\textrm{Tr} \left(|\hat{D}_{\alpha}|\right)=\int |D_{\alpha}(x,x;\kappa)| dx <\infty 
\label{sec3 : eq9a}
\end{eqnarray}
by employing \ref{sec2 : eq9} and condition \ref{sec3 : eq6a}.
\begin{description}
\item[($i$)] From the FGF \ref{sec2 : eq20a} we can easily identify the singular term, in the $\kappa \downarrow 0$ limit, to be of rank one and given by \ref{sec3 : eq8}. 
\item[($ii$)] The finite component of the Green's function is also given by
\begin{eqnarray}
G_{\alpha, \textrm{fin.}}(x,y;\kappa) &=& \frac{1}{2\pi}\int_{-\infty}^{\infty} \frac{(e^{-ip(x-y)}-1)}{K_{\alpha}|p|^{\alpha}+\kappa^2}dp. 
\label{sec3 : eq10}
\end{eqnarray}   
The second term in \ref{sec3 : eq10} is resposible for the singular part of the Green's function
\begin{eqnarray}
\frac{1}{2\pi}\int_{-\infty}^{\infty} \frac{dp}{K_{\alpha}|p|^{\alpha}+\kappa^2}= \frac{1}{\alpha (K_{\alpha})^{\frac{1}{\alpha}}\kappa^{2\left(1-\frac{1}{\alpha}\right)}\sin\left(\frac{\pi}{\alpha}\right)}
\label{sec3 : eq10a}
\end{eqnarray}
while the first is given by \ref{sec2 : eq20a}. Thus 
\begin{eqnarray}
\left|D_{\alpha, \textrm{fin.}}(x,y;\kappa) \right| \leq M(\alpha)\sqrt{|V(x)|}|x-y|^{\alpha-1}\sqrt{|V(y)|}  
\label{sec3 : eq11}
\end{eqnarray}
where 
\begin{eqnarray}
M(\alpha)=\frac{1}{2K_{\alpha}\Gamma(\alpha)}\frac{1}{|\cos\left(\frac{\alpha \pi}{2}\right)|}\geq M\left(2\right)=\frac{1}{2 K_{2}} 
\label{sec3 : eq11a}
\end{eqnarray}
since $M$ is a stricly decreasing function of $\alpha$. By definition an operator $\hat{T}$ belongs to the Hilbert-Schmidt class if  
\begin{eqnarray}
\parallel \hat{T} \parallel_{_{HS}}^2\equiv \textrm{Tr}(\hat{T}^* \hat{T})<\infty.
\label{sec3 : eq12}
\end{eqnarray}
We have
\begin{eqnarray}
\parallel \hat{D}_{\alpha, \textrm{fin.}} \parallel_{_{HS}}^2 &\leq& M^2(\alpha) \int \int |V(x)| |x-y|^{2(\alpha-1)}|V(y)|dx \, dy \non \\
&\leq& M^2(\alpha)\int \int |V(x)| (|x|+|y|)^{2(\alpha-1)}|V(y)|dx \, dy
\label{sec3 : eq13}
\end{eqnarray}
which by virtue of condition \ref{sec3 : eq6a} the double integral converges. It is worth noting that condition \ref{sec3 : eq6a} using the $c_p$-inequality 
\begin{eqnarray}
|f+g|^p\leq c_p(|f|^p+|g|^p), \quad c_p=\left\{\begin{array}{ll} 1 & \rm{if} \quad 0<p<1 \\ 2^{p-1} & \rm{for} \quad p\geq 1 \end{array} \right.
\label{sec3 : eq13a} 
\end{eqnarray}
can be written equivallently as
\begin{eqnarray}
\int (1+|x|^{2(\alpha-1)}) |V(x)| dx < \infty, \quad \alpha \in (1,2).
\label{sec3 : eq13b} 
\end{eqnarray}
\end{description}
Using \ref{sec3 : eq13} we observe that $\parallel g\hat{D}_{\alpha, \textrm{fin.}}\parallel_{_{HS}} < 1$. This condition allows the ground state to go to zero as $g\downarrow 0$ for $\alpha \in (\tilde{\alpha},2]$ with $\tilde{\alpha}$ being the unique solution of $M(\alpha)=1$. In this case  one can proceed following analogous steps as in \cite{Ref4} and study the existence and asymptotic behaviour of the ground energy in the small coupling constant limit. 
\section{Existence and asymptotic behaviour of the ground state}
\label{sec4}

\begin{theorem}
If condition \ref{sec3 : eq6a} is satisfied then for small $g$ there exists an eigenvalue of \ref{sec3 : eq01} if and only if equation
\begin{eqnarray}
\kappa^{2\left(1-\frac{1}{\alpha}\right)}=\frac{g}{\alpha \left(K_{\alpha}\right)^{\frac{1}{\alpha}}\sin\left(\frac{\pi}{\alpha}\right)}<\sqrt{|V|},(1+g\hat{D}_{\alpha, \textrm{fin.}})^{-1}\sqrt{|V|}>\equiv H(g,\kappa,\alpha)
\label{sec4 : eq1}
\end{eqnarray}
has a solution $\kappa^{2\left(1-\frac{1}{\alpha}\right)}>0$ and the eigenvalue is uniquely determined by $E=-\kappa^{4\left(1-\frac{1}{\alpha}\right)}$.
\label{sec4 : th1} 
\end{theorem}
\begin{proof}
\label{sec4 : pr1} 
\end{proof}
The operator $1+g\hat{D}_{\alpha, \textrm{fin.}}$ is inverible and from \ref{sec3 : eq3}, $-1$ is an eigenvalue of $\hat{D}_{\alpha}$ if and only if
\begin{eqnarray}
\det(1+g\hat{D}_{\alpha})=\det(1+g\hat{D}_{\alpha, \textrm{fin.}}) \det(1+(1+g\hat{D}_{\alpha, \textrm{fin.}})^{-1}g\hat{D}_{\alpha, \textrm{sing.}})=0.
\label{sec4 : eq2}
\end{eqnarray}
Equation \ref{sec4 : eq2} leads us to the conclusion that
\begin{eqnarray}
\det(1+(1+g\hat{D}_{\alpha, \textrm{fin.}})^{-1}g\hat{D}_{\alpha, \textrm{sing.}})=1+\textrm{Tr}((1+g\hat{D}_{\alpha, \textrm{fin.}})^{-1}g\hat{D}_{\alpha, \textrm{sing.}})=0.
\label{sec4 : eq3}
\end{eqnarray}
The first equality holds since $(1+\hat{D}_{\alpha, \textrm{fin.}})^{-1}\hat{D}_{\alpha, \textrm{sing.}}$ is of rank one. Equation \ref{sec4 : eq3} implies \ref{sec4 : eq1} which to order $g^2$ becomes
\begin{eqnarray}
\kappa^{2\left(1-\frac{1}{\alpha}\right)} &=& \frac{g}{\alpha \left(K_{\alpha}\right)^{\frac{1}{\alpha}}\sin\left(\frac{\pi}{\alpha}\right)} \parallel \sqrt{|V|}\parallel^2 \non \\
&+& \frac{g^2}{2 K_{\alpha}^{1+\frac{1}{\alpha}}\Gamma(1+\alpha) \sin\left(\frac{\pi}{\alpha}\right)\cos\left(\frac{\alpha \pi}{2}\right)}\int |V(x)||x-y|^{\alpha-1}|V(y)|dx dy +O(g^3).
\label{sec4 : eq4}
\end{eqnarray}
To prove the uniqueness of solutions we consider the existence of two solutions $\kappa_1$ and $\kappa_2$ satisfying equation \ref{sec4 : eq1}. Their distance is
\begin{eqnarray}
|\kappa_1^{2\left(1-\frac{1}{\alpha}\right)}-\kappa_2^{2\left(1-\frac{1}{\alpha}\right)}|=|H(g,\kappa_1,\alpha)-H(g,\kappa_2,\alpha)|\leq \int_{\kappa_1}^{\kappa_2}\left|\frac{\partial H(g,\kappa,\alpha)}{\partial \kappa}\right| d\kappa.
\label{sec4 : eq5}
\end{eqnarray} 
But 
\begin{eqnarray}
\left|\frac{\partial H(g,\kappa,\alpha)}{\partial \kappa}\right| &=& \left|\frac{g^2}{\alpha \left(K_{\alpha}\right)^{\frac{1}{\alpha}}\sin\left(\frac{\pi}{\alpha}\right)}<\sqrt{|V|},(1+g\hat{D}_{\alpha, \textrm{fin.}})^{-1} \frac{\partial \hat{D}_{\alpha, \textrm{fin.}}}{\partial \kappa} (1+g\hat{D}_{\alpha, \textrm{fin.}})^{-1} \sqrt{|V|}>\right| \non \\
&\leq& \frac{g^2}{\alpha \left(K_{\alpha}\right)^{\frac{1}{\alpha}}\sin\left(\frac{\pi}{\alpha}\right)}\left|<\sqrt{|V|},\frac{\partial \hat{D}_{\alpha, \textrm{fin.}}}{\partial \kappa}\sqrt{|V|}>\right| 
\label{sec4 : eq6}
\end{eqnarray}
An explicit calculation, employing \ref{sec2 : eq26}, shows that $\lim_{\kappa \downarrow 0}\left|<\sqrt{|V|},\frac{\partial \hat{D}_{\alpha, \textrm{fin.}}}{\partial \kappa}\sqrt{|V|}>\right|=0$ \footnote{Only for $\alpha=2$ this limit is bounded by 
\begin{displaymath}
\lim_{\kappa \downarrow 0}\left|<\sqrt{|V|},\frac{\partial \hat{D}_{2, \textrm{fin.}}}{\partial \kappa}\sqrt{|V|}>\right|\leq \frac{1}{4}\int |V(x)|(|x|^2+|y|^2) |V(y)|dx dy<\infty. 
\end{displaymath}}, thus we result in $\kappa_1=\kappa_2$.
\section{Conclusion}
\label{sec5}

\par We have shown that the fractional Weyl operator fulfills certain properties which are of practical importance if one is interested in studying the spectrum of perturbed fractional Schr\"{o}dinger operators. 
\par The next step was to obtain the small and long distance behaviour of the FGF. This was achieved by employing the Fourier-Laplace transform for the Green's function,  adopting a specific contour representation for the function $e^{-u^{\alpha}}$ and applying Cauchy's residue theorem. The asymptotic behaviour in the neighbourhood of the point $z=0$ reveals the terms related to the singular and finite parts of the integral kernel of the Birman-Schwinger operator. We have also determined the asymptotic behaviour of $\partial G_{\alpha}(x,y;\kappa)/\partial \kappa$ which is essential for proving the uniqueness of the ground state. 
\par Considering a Schr\"{o}dinger operator with kinetic term the Weyl operator and a small tunable potential which satisfies certain conditions, we were able to give a representaion for the Birman-Schwinger operator. This representation holds for all irrational values of the index-$\alpha$ in the interval $(1,2)$ and generalizes a previously known result for the first order Laplacian (the $\alpha=2$ case).
\par The existence of the lowest eigenvalue is related to the solution of equation \ref{sec4 : eq1} which is shown to be unique.
\par Finally, from the physical point of view there is a plethora of applications related to the fractional diffusion equation \cite{Ref1a} but not to the Schr\"{o}dinger equation. Only a year ago, the authors of \cite{Ref10} constructed a one-dimensional lattice model with a hopping particle and numerically obtained the eigenvalues and eigenfunctions in a bounded domain with different boundary conditions. In the continuum limit a quantum representation of this model could be a realization of the fractional Schr\"{o}dinger equation.

\addcontentsline{toc}{subsection}{Appendix A }
\section*{Appendix A }
\label{apA}
\renewcommand{\theequation}{A.\arabic{equation}}
\setcounter{equation}{0}

To prove \ref{sec1 : eq19} we first observe that
\begin{eqnarray}
\hat{\mathcal{P}}^{\alpha}e^{ipx}&=& \frac{1}{\cos(\frac{\pi \alpha}{2})}\frac{(ip)^{[\alpha]+1}}{\Gamma([\alpha]+1-\alpha)}e^{ipx}\int_{0}^{\infty}\frac{e^{-ipu}+(-1)^{[\alpha]+1}e^{ipu}}{u^{\alpha-[\alpha]}}du
\label{apA : eq1}
\end{eqnarray}
where the function $(\pm ip)^{\alpha}$ is to be understood as 
\begin{eqnarray}
(\pm ip)^{\alpha}=|p|^{\alpha}e^{\pm \frac{i\alpha \pi}{2} sign(p)}.
\label{apA : eq2}
\end{eqnarray}
\noi We distinguish two posibilities depending on whether $[\alpha]$ is even or odd.
\begin{enumerate}
\item[($\alpha$)] Substituting $[\alpha]=2l$ in \ref{apA : eq1} one obtains
\begin{eqnarray}
\hat{\mathcal{P}}^{\alpha}e^{ipx}&=& -\frac{2}{\cos(\frac{\pi \alpha}{2})}\frac{(-1)^{l+1}}{\Gamma(2l+1-\alpha)}|p|^{\alpha}e^{ipx}\int_{0}^{\infty}\frac{\sin(w)}{w^{\alpha-2l}}dw \non \\
&=& 2|p|^{\alpha}e^{ipx}.
\label{apA : eq3}
\end{eqnarray}
\item[($\beta$)] Setting $[\alpha]=2l+1$ in \ref{apA : eq1} yields
\begin{eqnarray}
\hat{\mathcal{P}}^{\alpha}e^{ipx}&=& \frac{2}{\cos(\frac{\pi \alpha}{2})}\frac{(-1)^{l+1}}{\Gamma(2l+2-\alpha)}|p|^{\alpha}e^{ipx}\int_{0}^{\infty}\frac{\cos(w)}{w^{\alpha-(2l+1)}}dw \non \\
&=& 2|p|^{\alpha}e^{ipx}.
\label{apA : eq4}
\end{eqnarray}
\end{enumerate}
We have used successively the integral formula [3.761.4] of \cite{Ref9}  
\begin{displaymath}
\int x^{\mu -1}\sin(zx)dx=\frac{\Gamma(\mu)}{z^{\mu}}\sin\left(\frac{\mu \pi}{2}\right), \quad z>0, \,\, 0<|\textrm{Re}\, \mu|<1 
\end{displaymath}
and \ref{sec2 : eq11a}.
\addcontentsline{toc}{subsection}{Appendix B }
\section*{Appendix B }
\label{apB}
\renewcommand{\theequation}{B.\arabic{equation}}
\setcounter{equation}{0}
The inversion formula \ref{sec2 : eq11} can be proved by applying Cauchy's residue theorem. Consider the rectangular contour $\mathcal{L}$ described in the positive sense (counterclockwise) with vertices $c\pm iR$, $c-\left(N+\frac{1}{2}\right)\alpha \pm iR$ where $N$ is a positive integer. The simple poles of $\Gamma\left(\frac{s}{\alpha}\right)$ lie in the interior of the contour at $0,-\alpha,-2\alpha, \cdots -N \alpha$. The residues of $f(s)=u^{-s}\Gamma\left(\frac{s}{\alpha}\right)$ are
\begin{eqnarray}
\textrm{Res}(f(s_j), s_j=-j\alpha)=\alpha \frac{(-1)^j}{\Gamma(j+1)}u^{\alpha j}, \quad j=0,1,\cdots,N.
\label{apB : eq1}
\end{eqnarray} 
Cauchy's theorem then gives
\begin{eqnarray}
\frac{1}{\alpha 2\pi i}\int_{\mathcal{L}}u^{-s} \Gamma\left(\frac{s}{\alpha}\right) ds=\sum_{j=0}^{N}\frac{(-1)^j}{\Gamma(1+j)}\left(u^{\alpha}\right)^j.
\label{apB : eq2}
\end{eqnarray}
Letting $R,N$ tend to infinity the integral on $\mathcal{L}$ minus the line joining $c-iR$ to $c+iR$ tends to zero. This is proved using Stirling's asymptotic formula 
\begin{eqnarray}
\Gamma(z)\approx \sqrt{2\pi} z^{z-\frac{1}{2}}e^{-z} \quad \textrm{as} \quad Rez\rightarrow \infty
\label{apB : eq3}
\end{eqnarray}
and 
\begin{eqnarray}
\left|\Gamma(x+iy)\right|=\sqrt{2\pi}|y|^{x-\frac{1}{2}}e^{-\pi \frac{|y|}{2}}\left[ 1+O\left(\frac{1}{|y|}\right)\right]
\label{apB : eq4}
\end{eqnarray}
when $x\in [x_1,x_2]$ and $|y|\rightarrow \infty$.

\bibliographystyle{plain}

\end{document}